		\documentclass[conference, 10pt]{IEEEtran}
		\newtheorem{theorem}{Theorem}

\newtheorem{proposition}{Proposition}

\newtheorem{problem}{Problem}

\newenvironment{proof}[1][Proof]{\begin{trivlist}
\item[\hskip \labelsep {\bfseries #1}]}{\end{trivlist}}

\newcommand{\qed}{\nobreak \ifvmode \relax \else
      \ifdim\lastskip<1.5em \hskip-\lastskip
      \hskip1.5em plus0em minus0.5em \fi \nobreak
      \vrule height0.75em width0.5em depth0.25em\fi}

\usepackage{color}
\usepackage{xcolor}
\usepackage{colortbl}
\definecolor{purple}{RGB}{139, 0, 139}
\usepackage{manfnt}
\usepackage{subfigure}
\usepackage{hyperref}

\newif\iftodo   
\todotrue
\newif\iftodoshort  
\todoshortfalse
\ifCLASSINFOpdf
  \usepackage[pdftex]{graphicx}
   \DeclareGraphicsExtensions{.pdf,.jpeg,.png}
\else
   \usepackage[dvips]{graphicx}
   \DeclareGraphicsExtensions{.eps}
\fi
\usepackage{amssymb,amsmath}
\usepackage[mathscr]{euscript}
\usepackage{bm} 
\usepackage{algorithmic}
\usepackage[ruled,vlined,commentsnumbered]{algorithm2e}
\usepackage{cite}
\graphicspath{ {./fig/} }
\newcommand{\Rmnum}[1]{\uppercase\expandafter{\romannumeral #1}}

\newcommand{\diag}{\mathop{\mathrm{diag}}}

\newcommand{\field}[1]{\mathbb{#1}}

\newcommand{\emenge}[1]{\mathscr{#1}}
\newcommand{\set}[1]{\mathscr{#1}}

\newcommand{\operator}[1]{\mathrm{#1}}

\newcommand{\R}{{\field{R}}}

\newcommand{\dl}{(\text{d})}
\newcommand{\ul}{(\text{u})}
\newcommand{\ula}{(\text{u},1)}
\newcommand{\ulb}{(\text{u},2)}

\newcommand{\Ns}{{\emenge{N}}}
\newcommand{\Ks}{{\emenge{K}}}
\newcommand{\Cs}{{\emenge{C}}}

\newcommand{\sinr}{\operator{SINR}}

\newcommand{\ma}{\bm}
\newcommand{\ve}{\bm}

\newcommand{\V}{\tilde{\ma{V}}}
\newcommand{\Pm}{P^{\text{max}}}
\newcommand{\hV}{\hat{\ma{V}}}
\newcommand{\hW}{\hat{\ma{W}}}
\newcommand{\Aa}{\ma{A}_{\ve{\alpha}}}
\newcommand{\Bb}{\ma{B}_{\ve{\beta}}}
\newcommand{\Vb}{\V_{\ve{b}}}
\newcommand{\Vt}{\V_{\ve{\theta}}}

\newcommand{\uGb}{\underline{\ma{G}}_{\ve{b}}}
\newcommand{\oGb}{\overline{\ma{G}}_{\ve{b}}}
\newcommand{\uJc}{\set{J}_c^{\ul}}
\newcommand{\uJca}{\set{J}_c^{(\text{u},1)}}
\newcommand{\uJcb}{\set{J}_c^{(\text{u},2)}}

\newcommand{\uJn}{\widehat{\set{J}}_n^{\ul}}
		\newcommand{\cosl}[1]{}
		\newcommand{\insl}[1]{#1}
		\newcommand{\resl}[1]{}

\begin{document}
		%
		\title{Joint Optimization of Coverage, Capacity and Load Balancing in Self-Organizing Networks}
		\author{
		    \IEEEauthorblockN{Qi Liao\IEEEauthorrefmark{1}, Daniyal Amir Awan\IEEEauthorrefmark{1}, S\l awomir Sta\'{n}czak\IEEEauthorrefmark{1}\IEEEauthorrefmark{2}}

		    \IEEEauthorblockA{\IEEEauthorrefmark{1}Technische Universit\"{a}t Berlin, 10587 Berlin, Germany \\
		    qi.liao.de@ieee.de, \{daniyal.a.awan, slawomir.stanczak\}@tu-berlin.de}
		}
		\maketitle
		
		\begin{abstract}
		  This paper develops an optimization framework for self-organizing
		  networks (SON). The objective is to ensure efficient network
		  operation by a joint optimization of different SON functionalities,
		  which includes capacity, coverage and load balancing. Based on the
		  axiomatic framework of monotone and strictly subhomogeneous function, we formulate
		  an optimization problem for the uplink and propose a two-step
		  optimization scheme using fixed point iterations: i) per base station antenna tilt optimization
		  and power allocation, and ii) cluster-based base station assignment
		  of users and power allocation.  We then consider the downlink, which
		  is more difficult to handle due to coupled variables, and show
		  downlink-uplink duality relationship. As a result, a solution for
		  the downlink is obtained by solving the uplink problem.  Simulations
		  show that our approach achieves a good trade-off between coverage,
		  capacity and load balancing.
		\end{abstract}
		%
		
		\section{Introduction}\label{sec:intro}
		
		A major challenge towards self-organizing networks (SON) is the joint
		optimization of multiple SON use cases by
		coordinately handling multiple configuration parameters. Widely
		studied SON use cases include coverage and capacity optimization
		(CCO), mobility load balancing (MLB) and mobility robustness
		optimization (MRO)\cite{3GPP36902}.\cosl{We need a reference here} However, most of
		these works study an isolated single use case and ignore the conflicts
		or interactions between the use cases
		\cite{giovanidis2012dist,razavi2010self}. 
		
		In contrast, this paper considers a joint optimization of two strongly
		coupled use cases: CCO and MLB. The objective is to achieve a good
		trade-off between coverage and capacity performance, while ensuring a
		load-balanced network. The SON functionalities are usually implemented
		at the network management layer and are designed to deal with \lq\lq
		long-term\rq\rq \ network performance. Short-term optimization of
		individual users is left to lower layers of the protocol stack. To
		capture long-term global changes in a network, we consider a
		cluster-based network scenario, where users served by the same base
		station (BS) with similar SINR distribution are adaptively grouped
		into clusters. Our objective is to jointly optimizing the following
		variables:
		\begin{itemize}
		\item Cluster-based BS assignment and power allocation.
		\item BS-based antenna tilt optimization and power allocation.
		\end{itemize}
		The joint optimization of assignment, antenna tilts, and powers is an
		inherently challenging problem. The interference and the resulting
		performance measures depend on these variables in a complex and
		intertwined manner. Such a problem, to the best of the authors'
		knowledge, has been studied in only a few works. For example, in
		\cite{klessig2012improving} a problem of jointly optimizing antenna
		tilt and cell selection to improve the spectral and energy efficiency
		is stated, however, the solution derived by a structured searching
		algorithm may not be optimal.
		
		In this paper, we propose a robust algorithmic framework built on a
		utility model, which enables fast and near-optimal uplink solutions and sub-optimal downlink solutions\cosl{Do
		  we know that this is near-optimal?} by exploiting three properties:
		1) the monotonic property and fixed point of the monotone and strictly subhomogenoues (MSS) functions \footnote{Many literatures use the term {\it interference function} for the functions satisfy three condotions, positivity, monotonicity and scalability \cite{yates95}. Positivity is shown to be a consequence of the other two properties \cite{leung2004convergence}, and we use the term {\it strctly subhomogeneous} in place of scalable from a constraction mapping point of view in keeping with some related literature \cite{nuzman2007contraction}.}, 2)
		decoupled property of the antenna tilt and BS assignment optimization
		in the uplink network, and 3) uplink-downlink duality. The first
		property admits global optimal solution with fixed-point iteration for
		two specific problems: utility-constrained power minimization and
		power-constrained max-min utility balancing
		\cite{vucic2011fixed,stanczak2009fundamentals,schubert2012interference,yates95}. The
		second and third properties enable decomposition of the
		high-dimensional optimization problem, such as the joint beamforming
		and power control proposed in
		\cite{BocheDuality06,schubert2005iterative,huang2013joint,he2012multi}. Our
		distinct contributions in this work can be summarized as follows:\\
		1) We propose a max-min utility balancing algorithm for
		capacity-coverage trade-off optimization over a joint space of antenna
		tilts, BS assignments and powers. The utility defined as a convex
		combination of the average SINR and the worst-case SINR implies the balanced performance of capacity and coverage. Load
		balancing is improved as well due to a uniform distribution of the
		interference among the BSs.\\
		2) The proposed utility is formulated based on the MSS functions, which allows us to find the optimal solution by applying
		fixed-point iterations.\\
		3) Note that antenna tilts are BS-specific variables, while assignments are cluster-specific, we develop two optimization problems with the same objective functions,  formulated either as a problem of per-cluster variables or as a problem of per-base variables. We
		propose a two-step optimization algorithm in the uplink to iteratively
		optimize the per BS variables (antenna tilts and BS power budgets) and the cluster-based variables (assignments and cluster power). Since both problems aim at optimizing the same objective function, the algorithm is shown to be convergent.\\
		4) The decoupled property of antenna tilt and assignment in the uplink decomposes the high-dimensional optimization problem and enables more efficient optimization algorithm. We then analyze the uplink-downlink duality by using the Perron-Frobenius theory\cite{meyer2000matrix}, and propose an efficient
		sub-optimal solution in the downlink by utilizing optimized variables
		in the dual uplink.
		   
		
		\section{System Model}\label{sec:Model}
		
		We consider a multicell wireless network composed of a set of BSs
		$\set{N}:=\{1,\ldots, N\}$ and a set of users $\set{K}:=\{1,\ldots,
		K\}$. Using fuzzy C-means clustering algorithm \cite{bezdek1984fcm},
		we group users with similar SINR distributions\footnote{We assume the
		  Kullback-Leibler divergence as the distance metric.} and served by the same BS into
		clusters. The clustering algorithm is beyond the scope of this
		paper. Let the set of user clusters be denoted by
		$\set{C}:=\{1,\ldots,C\}$, and let $\ma{A}$ denote a $C\times K$
		binary user/cluster assignment matrix whose columns sum to one.  The
		BS/cluster assignment is defined by a $N\times C$ binary matrix
		$\ma{B}$ whose columns also sum to one.  
		
		Throughout the paper, we assume a frequency flat channel. The
		average/long-term downlink path attenuation between $N$ BSs and $K$
		users are collected in a channel gain matrix $\ma{H}\in \R^{N\times
		  K}$.  We introduce the cross-link gain matrix $\ma{V}\in\R^{K\times
		  K}$, where the entry $v_{lk}(\theta_j)$ is the cross-link gain
		between user $l$ served by BS $j$, and user $k$ served by BS $i$,
		i.e., between the transmitter of the link $(j, l)$ and the receiver of
		the link $(i, k)$. Note that $v_{lk}(\theta_j)$ depends on the antenna
		downtilt $\theta_j$.  Let the BS/user assignment matrix be denoted by
		$\ma{J}$ so that we have $\ma{J}:=\ma{B}\ma{A}\in\{0,1\}^{N\times K}$,
		and $\ma{V}:=\ma{J}^T\ma{H}$.  We denote by $\ve{r}:=[r_1, \ldots,
		r_N]^T$, $\ve{q}:=[q_1, \ldots, q_C]^T$ and $\ve{p}:=[p_1, \ldots,
		p_K]^T$the BS transmission power budget, the cluster power allocation
		and the user power allocation, respectively.
		%
		
		\subsection{Inter-cluster and intra-cluster power sharing factors}
		\label{subsec:powFactor}
		
		We introduce the inter-cluster and intra-cluster power sharing factors
		to enable the transformation between two power vectors with different
		dimensions.  Let $\ve{b}:=[b_1, \ldots, b_C]^T$ denote the serving BSs
		of clusters $\{1, \ldots, C\}$. We define the vector of the
		inter-cluster power sharing factors to be $\ve{\beta}:=[\beta_1,
		\ldots, \beta_C]^T$, where $\beta_c:=q_c/r_{b_c}$.  With the
		BS/cluster assignment matrix $\ma{B}$, we have $\ve{q}:=\Bb^T \ve{r}$,
		where $\Bb:=\ma{B}\diag\{\ve{\beta}\}$.  Since users belonging to the
		same cluster have similar SINR distribution, we allocate the cluster
		power uniformly to the users in the cluster. The intra-cluster sharing
		factors are represented by $\ve{\alpha}:=[\alpha_1, \ldots,
		\alpha_K]^T$ with $\alpha_k=1/|\set{K}_{c_k}|$ for $k\in\set{K}$,
		where $\set{K}_{c_k}$ denotes the set of users belonging to cluster
		$c_k$, while $c_k$ denotes the cluster with user $k$. We have
		$\ve{p}:=\Aa^T\ve{q}$, where $\Aa:=\ma{A}\diag\{\ve{\alpha}\}$. The
		transformation between BS power $\ve{r}$ and user power $\ve{p}$ is
		then $\ve{p}:=\ma{T}\ve{r}$ where the transformation matrix
		$\ma{T}:=\Aa^T\Bb^T$.
		%
		\subsection{Signal-to-interference-plus-noise ratio}\label{subsec:SINR}
		
		Given the cross-link gain matrix $\ma{V}$, the downlink SINR of the $k$th user depends on all
		 powers and is given by
		\begin{equation}
		\sinr_k^{\dl}:=\frac{p_k \cdot v_{kk}(\theta_{n_k})}{\sum_{l\in\set{K}\setminus k} p_l \cdot v_{lk}(\theta_{n_l})+\sigma_k^2}, k\in\set{K} 
		\label{eqn:DL_SINR}
		\end{equation}
		where $n_k$ denotes the serving BS of user $k$, $\sigma_k^2$ denotes
		the noise power received in user $k$. Likewise, the uplink SINR is
		\begin{equation}
		\sinr_k^{\ul}:=\frac{p_k \cdot v_{kk}(\theta_{n_k})}{\sum_{l\in\set{K}\setminus k} p_l \cdot v_{kl}(\theta_{n_k})+\sigma_k^2}, k\in\set{K} 
		\label{eqn:UL_SINR}
		\end{equation}
		%
		Assuming that there is no self-interference, the cross-talk terms can
		be collected in a matrix
		\begin{equation}
		  [\V]_{lk}:=
		  \begin{cases}
		    v_{lk}(\theta_{n_l}), & l\neq k\\
		    0, & l=k
		  \end{cases}.
		  \label{eqn:PsiMat}
		\end{equation} 
		Thus the downlink interference received by user $k$ can be written as
		$I_k^{\dl}:=[\tilde{\ma{V}}^T\ve{p}]_k$, while the uplink interference
		is given by $I_k^{\ul}:=[\tilde{\ma{V}}\ve{p}]_k$.
		
		A crucial property is that the uplink SINR of user $k$ depends on the
		BS assignment $n_k$ and the single antenna tilt $\theta_{n_k}$ alone,
		while the downlink SINR depends on the BS assignment vector
		$\ve{n}:=[n_1,\ldots, n_K]^T$, and the antenna tilt vector
		$\ve{\theta}:=[\theta_1, \ldots, \theta_N]^T$. The decoupled property
		of uplink transmission has been widely exploited in the context of
		uplink and downlink multi-user beamforming \cite{BocheDuality06}\cosl{Reference} and
		provides a basis for the optimization algorithm in this paper. 

		The notation used in this paper is summarized in Table \ref{tab:CovCap_notation}.

\begin{table}[t]
\centering
\caption{NOTATION SUMMARY}
\begin{tabular}{|c|c|}
\hline
$\Ns$ & set of BSs  \\
$\Ks$ & set of users \\
$\Cs$ & set of user clusters\\ 
$\ma{A}$ & cluster/user assignment matrix\\
$\ma{B}$ & BS/cluster assignment matrix\\
$\ma{J}$ & BS/user assignment matrix\\
$c_k$ & cluster that user $k$ is subordinated to\\
$\Ks_{c}$ & set of users subordinated to cluster $c$\\
$\ma{H}$ & channel gain matrix\\
$\ma{V}$ & interference coupling matrix\\
$\tilde{\ma{V}}$ & interference coupling matrix without intra-cell interference\\
$\tilde{\ma{V}}_{\ve{b}}$ & interference coupling matrix depending on BS assignments $\ve{b}$\\
$\tilde{\ma{V}}_{\ve{\theta}}$ & interference coupling matrix depending on antenna tilts $\ve{\theta}$\\
$\ve{r}$ & BS power budget vector\\
$\ve{q}$ & cluster power vector\\
$\ve{p}$ & user power vector\\
$\ve{\alpha}$ & intra-cluster power sharing factors\\
$\ve{\beta}$ & inter-cluster power sharing factors\\
$\ma{A}_{\ve{\alpha}}$ & transformation from $\ve{q}$ to $\ve{p}$, $\ve{p}:=\ma{A}_{\ve{\alpha}}^T\ve{q}$\\
$\ma{B}_{\ve{\beta}}$ & transformation from $\ve{r}$ to $\ve{q}$, $\ve{q}:=\ma{B}_{\ve{\beta}}^T\ve{r}$\\
$\ma{T}$ & transformation from $\ve{r}$ to $\ve{p}$, $\ve{p}:=\ma{T}\ve{r}$\\ 
$\ve{\theta}$ & BS antenna tilt vector\\
$\ve{b}$ & serving BSs of clusters\\
$b_c$ & serving BS of cluster $c$\\
$\ve{n}$ & serving BSs of the users\\
$n_k$ & serving BS of user $k$\\
$\ve{\sigma}$ & noise power vector\\
$\Pm$ & sum power constraint\\
\hline
\end{tabular}
\label{tab:CovCap_notation}
\end{table}

		\section{Utility Definition and Problem Formulation}\label{sec:ProbForm}
		
		As mentioned, the objective is a joint optimization of coverage,
		capacity and load balancing. We capture coverage by the worst-case
       SINR, while the average SINR is used to represent capacity. A cluster-based utility $U_c(\ve{\theta},\ve{r},\ve{q},\ve{b})$ is introduced as the combined function of the worst-case SINR and average SINR, depending on BS
		power allocation $\ve{r}$, antenna downtilt $\ve{\theta}$ ,
		cluster power allocation $\ve{q}$ and BS/cluster assignment
		$\ve{b}$.\footnote{The reader should note that user-specific variables
		  $(\ve{p},\ve{n})$ can be derived directly from cluster-specific
		  variables $\ve{q}$ and $\ve{b}$, provided that cluster/user
		  assignment $\ma{A}$ and intra-cluster power sharing factor
		  $\ve{\alpha}$ are given.}	 To achieve the load balancing by distributing the clusters to the BSs such that their utility targets can be achieved \footnote {The assignment of clusters also distributes the interference among the BSs.}, we formulate the following objective
		$$\max_{(\ve{r},\ve{\theta},\ve{q},\ve{b})}\min_{c\in\set{C}} \frac{U_c(\ve{r},\ve{\theta},\ve{q},\ve{b})}{\gamma_c}$$
		where  $\gamma_c$ is the predefined utility target for cluster $c$.
The BS variables $(\ve{r},\ve{\theta})$ and cluster variables $(\ve{q}, \ve{b})$ are optimized by iteratively solving\\ 
1) Cluster-based BS assignment and power allocation
$\max_{(\ve{q},\ve{b})}\max_{c\in\set{C}} U_c(\ve{q},\ve{b})/\gamma_c$ given the fixed $(\hat{\ve{r}},\hat{\ve{\theta}})$ \\
2) BS-based antenna tilt optimization and power allocation $\max_{(\ve{r},\ve{\theta})}\max_{c\in\set{C}} U_c(\ve{r},\ve{\theta})/\gamma_c$ given the fixed $(\hat{\ve{q}},\hat{\ve{b}})$.

In the following we introduce the utility definition and problem formulation for the cluster-based and the BS-based problems respectively. We start with the problem statement and algorithmic approaches for the
		uplink. We then discuss the downlink in Section \ref{sec:Duality}.
		
		%
		\subsection{Cluster-Based BS Assignment and Power Allocation}\label{subsec:clusterOpt}
		
		Assume the per-BS variables
		$(\hat{\ve{r}}, \hat{\ve{\theta}})$ are fixed, let the interference
		coupling matrix depending on BS assignment $\ve{b}$ in
		\eqref{eqn:PsiMat} be denoted by $\Vb$. We first define two utility
		functions indicating capacity and coverage per cluster respectively,  then we introduce the joint utility as a combination of the capacity and coverage utility. After that we define the cluster-based max-min utility balancing problem based on the joint utility.
		%
		\subsubsection{Average SINR Utility (Capacity)}\label{subsubsec:LB_A}
		
		With the intra-cluster power sharing factor introduced in Section
		\ref{subsec:powFactor}, we have $\ve{p}:=\Aa^T \ve{q}$. Define the
		noise vector $\ve{\sigma}:=[\sigma_1^2, \ldots, \sigma_K^2]^T$, the
		average SINR of all users in cluster $c$ is written as
		\begin{align}
		\bar{U}_c^{\ula}&(\ve{q}, \ve{b})  := \frac{1}{|\set{K}_c|} \sum_{k\in\set{K}_c}\sinr_k^{\ul}\nonumber\\
		&= \frac{1}{|\set{K}_c|}  \sum_{k\in\set{K}_c}\frac{q_c \alpha_k v_{kk}}{\left[\Vb \Aa^T \ve{q}+\ve{\sigma}\right]_k}\nonumber\\
		&\geq \frac{1}{|\set{K}_c|}\frac{q_c \sum_{k\in\set{K}_c} \alpha_k v_{kk}}{\sum_{k\in\set{K}_c} \left[\Vb \Aa^T \ve{q}+\ve{\sigma}\right]_k} 
		=U_c^{\ula}(\ve{q}, \ve{b})
		\label{eqn:CL_cap_1}
		\end{align}
		The uplink capacity utility of cluster $c$ denoted by $U_c^{\ula}$ is
		measured by the ratio between the total useful power and the total
		interference power received in the uplink in the cluster. Utility
		$U_c^{\ula}$ is used instead of $\bar{U}_c^{\ula}$ because of two
		reasons: First, it is a lower bound for the average SINR. Second, it
		has certain monotonicity properties (introduced in Section
		\ref{sec:OPAlgor}) which are useful for optimization.
		
		Introducing the cluster coupling term  $\oGb^{\ul}:=\ma{\Psi}\ma{A}\Vb\Aa^T$, where $\ma{\Psi}:=\diag\{|\set{K}_1|/g_1, \ldots, |\set{K}_c|/g_C\}$ and $g_c:=\sum_{k\in \set{K}_c}\alpha_k v_{kk}$ for $c\in\set{C}$; and the noise term $\overline{\ve{z}}:=\ma{\Psi}\ma{A}\ve{\sigma}$, 
		 the capacity utility is simplified as 
		\begin{align}
		U_c^{\ula}(\ve{q}, \ve{b})&:=\frac{q_c}{\uJca(\ve{q}, \ve{b})}\label{eqn:CL_cap_2}\\
		\mbox{where } \uJca(\ve{q}, \ve{b})&:=\left[\oGb^{\ul}\ve{q}+\overline{\ve{z}}\right]_c. \label{eqn:CL_cap_inter}
		\end{align}
		%
		\subsubsection{Worst-Case SINR Utility (Coverage)}
		Roughly speaking, the coverage problem arises when a certain number of the SINRs are lower than the predefined SINR threshold. Thus, improving the coverage performance is equivalent to maximizing the worst-case SINR such that the worst-case SINR achieves the desired SINR target. We then define the uplink coverage utility for each cluster as
		\begin{align}
		U_c^{\ulb}(\ve{q},\ve{b})&:=\min_{k\in\set{K}_c}\sinr_k^{\ul}=\min_{k\in\set{K}_c} 
		                            \frac{q_c\alpha_k v_{kk}}{\left[\Vb \Aa^T \ve{q}+\ve{\sigma}\right]_k}\nonumber\\
															&= \frac{q_c}{\max_{k\in\set{K}_c}\left[ \ma{\Phi}\Vb \Aa^T \ve{q}+\ma{\Phi}\ve{\sigma}\right]_k}
		\label{eqn:CL_cov_1}
		\end{align}
		where $\ma{\Phi}:=\diag\{1/\alpha_1 v_{11}, \ldots, 1/\alpha_K v_{KK}\}$. We define a $C \times K$ matrix $\ma{X}:=[\ve{x}_1|\ldots|\ve{x}_C]^T$, where $\ve{x}_c:=\ve{e}^j_K$ and $\ve{e}^j_i$ denotes an $i$-dimensional binary vector which has exact one entry (the j-th entry) equal to 1. Introducing the term $\uGb^{\ul}:=\ma{\Phi}\Vb \Aa^T$, and the noise term $\underline{\ve{z}}:=\ma{\Phi}\ve{\sigma}$, the coverage utility is given by
		\begin{align}
		U_c^{\ulb}(\ve{q},\ve{b})&:=\frac{q_c}{\uJcb(\ve{q}, \ve{b})}\label{eqn:CL_cov_2}\\
		\mbox{where } \uJcb(\ve{q}, \ve{b}) & := \max_{\ve{x}_c:=\ve{e}_K^j, j\in\set{K}_c} \left[\ma{X}\uGb^{\ul}\ve{q}+\ma{X}\underline{\ve{z}}\right]_c. \label{eqn:CL_cov_inter}
		\end{align}
		%
		\subsubsection{Joint Utility and Cluster-Based Max-Min Utility Balancing}\label{eqn:LB_maxmin}
		The joint utility $U_c^{\ul}(\ve{q}, \ve{b})$ is defined as 
		\begin{align}
		U_c^{\ul}(\ve{q}, \ve{b})&:=\frac{q_c}{\uJc(\ve{q}, \ve{b})}\label{eqn:LB_utility_1}\\
		\mbox{where }\uJc(\ve{q}, \ve{b})&:= \mu\uJca(\ve{q}, \ve{b})+(1-\mu)\uJcb(\ve{q}, \ve{b})\label{eqn:LB_utility_2}.
		\end{align}
		In other words, the joint interference function $\set{I}_c^{\ul}$ is a convex combination of $\set{I}_c^{\ula}$ in \eqref{eqn:CL_cap_inter} and $\set{I}_c^{\ulb}$ in \eqref{eqn:CL_cov_inter}. 
		
		 The cluster-based power-constrained max-min utility balancing problem in the uplink is then provided by
		\begin{problem}[Cluster-Based Utility Balancing]
		\begin{equation}
		C^{\ul}(\Pm)=\max_{\ve{q}\geq 0, \ve{b}\in \set{N}^C} \min_{c\in\set{C}} \frac{U_c^{\ul}(\ve{q}, \ve{b})}{\gamma_c}, \mbox{s.t. } \|\ve{q}\|\leq \Pm
		\label{eqn:LB_OP}
		\end{equation}
		Here, $\|\cdot\|$ is an arbitrary monotone norm, i.e., $\ve{q}\leq\ve{q}'$ implies $\|\ve{q}\|\leq\|\ve{q}'\|$,  $\Pm$ denotes the total power constraint. 
		
		According to the joint utility in \eqref{eqn:LB_utility_1},\eqref{eqn:LB_utility_2}, the algorithm optimizes the performance of capacity when we set the tuning parameter $\mu=1$ (utility is equivalent to the capacity utility in \eqref{eqn:CL_cap_2}), while with $\mu=0$ it optimizes the performance of coverage (utility equals to the coverage utility in \eqref{eqn:CL_cov_2}). By tuning $\mu$ properly, we can achieve a good trade-off between the performance of coverage and capacity.
		\label{prob:LB}
		\end{problem} 
		%
		\subsection{BS-Based Antenna Tilt Optimization and Power Allocation}\label{subsec:AO}
		Given the fixed $(\hat{\ve{q}},\hat{\ve{b}})$, we compute the intra-cluster power allocation factor $\ve{\beta}$, given by $\beta_c:=\hat{q}_c/\sum_{c\in\set{C}_{b_c}}\hat{q}_c$ for $c\in\set{C}$. We denote the cross-link coupling matrix depending on $\ve{\theta}$ by $\Vt$. In the following we formulate the BS-based max-min utility balancing problem such that it has the same physical meaning as the problem stated in \eqref{eqn:LB_OP}. We then introduce the BS-based joint utility interpreted by $(\ve{r}, \ve{\theta})$.
		
		\subsubsection{BS-Based Max-Min Utility Balancing}\label{subsubsec:AO_maxmin}
		To be consistent with our objective function $C^{\ul}(\Pm)$ in \eqref{eqn:LB_OP}, we transform the cluster-based optimization problem to the BS-based optimization problem: 
		\begin{problem}[BS-Based Utility Balancing]
		\begin{align}
		C^{(u)}&(\Pm)=\max\limits_{\ve{r}\geq 0, \ve{\theta}\in\Theta^N} \min\limits_{c\in\set{C}}
		\frac{U_c^{\ul}(\ve{r},\ve{\theta})}{\gamma_c}\nonumber\\
		&=\max\limits_{\ve{r}\geq 0, \ve{\theta}\in\Theta^N}
		\min\limits_{n\in\set{N}}\left(\min\limits_{c\in\set{C}_n}\frac{U_c^{\ul}(\ve{r},\ve{\theta})}{\gamma_c}\right)\nonumber\\
		& = \max\limits_{\ve{r}\geq 0, \ve{\theta}\in\Theta^N} \min\limits_{n\in\set{N}}
		\widehat{U}_n^{\ul}(\ve{r},\ve{\theta}), \mbox{ s.t. } \|\ve{r}\|\leq P^{\text{max}}
		\label{eqn:maxmin_AO}
		\end{align}
		\label{prob:AO}
		\end{problem}
		where $\Theta$ denotes the predefined space for antenna tilt configuration.
		\subsubsection{BS-Based Joint Utility}\label{subsubsec:AO_joinyUtility}
		 It is shown in \eqref{eqn:maxmin_AO} that the cluster-based problem is transformed to the BS-based problem by defining 
		\begin{align}
		\widehat{U}_n^{\ul}(\ve{r},\ve{\theta})&:=\min_{c\in\set{C}_n}\frac{U_c^{\ul}(\ve{r},\ve{\theta})}{\gamma_c}= \frac{r_n}{\uJn(\ve{r}, \ve{\theta})}\label{eqn:AO_utility_1}\\
		\uJn(\ve{r}, \ve{\theta}) &:= \max_{c\in\set{C}_n} \frac{\gamma_c}{\beta_c} \uJc(\ve{r}, \ve{\theta}),
		\label{eqn:AO_utility_2}
		\end{align} 
		where $\uJc(\ve{r}, \ve{\theta})$ is obtained from $\uJc(\ve{q}, \ve{b})$ in \eqref{eqn:LB_utility_2} by substituting $\ve{q}$ with $\ve{q}:=\Bb^T\ve{r}$, and $\V_{\ve{b}}$ with $\V_{\ve{\theta}}$. Note that \eqref{eqn:AO_utility_1}  is derived by applying the inter-cluster sharing factor such that $r_n:=q_c/\beta_c$ for $n=b_c$. Due to lack of space we omit the details of the individual per BS capacity and coverage utilities corresponding to the cluster-based utilities \eqref{eqn:CL_cap_1} and \eqref{eqn:CL_cov_1}.
		\section{Optimization Algorithm}\label{sec:OPAlgor}
		We developed our optimization algorithm based on the fixed-point iteration algorithm proposed by Yates \cite{yates95}, by exploiting the properties of the monotone and strictly subhomogeneous functions.
		\subsection{MSS function and Fixed-Point Iteration}\label{subsec:contraction}
	  The vector function $\ve{f}: \R_+^K\mapsto \R_+^K$ of interest has the following two properties:
	\begin{itemize}
	\item {\it Monotonicity}:  $\ve{x}\leq \ve{y}$ implies $\ve{f}(\ve{x})\leq\ve{f}(\ve{y})$,.
	\item  {\it Strict subhomogeneity}: for each $\alpha>1, \ve{f}(\alpha \ve{x})<\alpha\ve{f}(\ve{x})$. 
	\end{itemize}
	A function satisfying the above two properties is referred to be  {\it monotonic and strict subhomogeneous (MSS)}. When the strict inequality is relaxed to weak inequality, the function is said to be {\it monotonic and subhomogeneous (MS)}.
	\begin{theorem}\cite{nuzman2007contraction}
	Suppose that $\ve{f}: \R_+^K\mapsto \R_+^K$ is MSS and that $\ve{h}=\ve{x}/l(\ve{x})$, where $l:\R_+^K \mapsto \R_+$ is MS. For each $\theta>0$, there is exactly one eigenvector $\ve{v}$ and the associated eigenvalue $\lambda$ of $\ve{f}$ such that $l(\ve{v})=\theta$. Given an arbitrary $\theta$, the repeated iterations of the function 
	\begin{equation}
	\ve{g}(\ve{x})=\theta \ve{f}(x)/l(\ve{f(x)})
	\label{eqn:fixedpointiteration}
	\end{equation}
	converge to a unique fixed point such that $l(\ve{v})=\theta$.
	\label{Theoremmapping}
	\end{theorem}
 The fixed point iteration in \eqref{eqn:fixedpointiteration} is used to obtain the solution of the following max-min utility balancing problem 
 \begin{equation}
 \max_{\ve{p}}\min_{k\in\set{K}} U_k(\ve{p}), \mbox{ s.t. } \|\ve{p}\|\leq P^{\text{max}}
 \label{eqn:prob_maxmin_1}
 \end{equation}
 where the utility function can be defined as  $U_k(\ve{p}):= p_k/f_k(\ve{p})$.

		\subsection{Joint Optimization Algorithm}\label{subsec:JointOptAlgor}
		We aim on jointly optimizing both problems, by optimizing $(\ve{q}, \ve{b})$ in Problem \ref{prob:LB} and $(\ve{r},\ve{\theta})$ in Problem \ref{prob:AO} iteratively with the fixed-point iteration. In the following we present some properties that are required to solve the problem efficiently and to guarantee the convergence of the algorithm. 
		\subsubsection{Decoupled Variables in Uplink}
		In uplink the variables $\ve{b}$ and $\ve{\theta}$ are decoupled in the interference functions \eqref{eqn:LB_utility_2} and \eqref{eqn:AO_utility_2}, i.e., $\uJc(\ve{q}, \ve{b}):=\uJc(\ve{q}, b_c)$ and $\uJn(\ve{r}, \ve{\theta}):=\uJn(\ve{r}, \theta_n)$. Thus, we can decompose the BS assignment (or tilt optimization) problem into sub-problems that can be independently solved in each cluster (or BS), and the interference functions can be modified as functions of the power allocation only:
		\begin{align}
		\uJc(\ve{q})&:=\min_{b_c\in\set{N}} \uJc(\ve{q}, b_c)\label{eqn:modi_inter_1}\\
		\uJn(\ve{r})&:=\min_{\theta_n\in\Theta} \uJn(\ve{r}, \theta_n) \label{eqn:modi_inter_2}
		\end{align} 
		\subsubsection{Standard Interference Function}
		The modified interference function \eqref{eqn:modi_inter_1} and \eqref{eqn:modi_inter_2} are \textit{standard}.
		Using the following three properties: 1) an affine function $\ve{\set{I}}(\ve{p}):=\ma{V}\ve{p}+\ve{\sigma}$ is standard, 2) if $\ve{\set{I}}(\ve{p})$ and $\ve{\set{I}}'(\ve{p})$ are standard, then $\beta\ve{\set{I}}(\ve{p})+(1-\beta)\ve{\set{I}}'(\ve{p})$ are standard, and 3) If $\ve{\set{I}}(\ve{p})$ and $\ve{\set{I}}'(\ve{p})$ are standard, then $\ve{\set{I}}^{\text{min}}(\ve{p})$ and $\ve{\set{I}}^{\text{max}}(\ve{p})$ are standard, where $\ve{\set{I}}^{\text{min}}(\ve{p})$ and $\ve{\set{I}}^{\text{max}}(\ve{p})$ are defined as $\set{I}_j^{\text{min}}(\ve{p}):=\min\{\set{I}_j(\ve{p}), \set{I}_j'(\ve{p})\}$ and $\set{I}_j^{\text{max}}(\ve{p}):=\max\{\set{I}_j(\ve{p}), \set{I}_j'(\ve{p})\}$ respectively \cite{yates95}, we can easily prove that \eqref{eqn:modi_inter_1} and \eqref{eqn:modi_inter_2} are standard interference functions.
		
		Substituting \eqref{eqn:modi_inter_1} and \eqref{eqn:modi_inter_2} in Problem \ref{prob:LB} and Problem \ref{prob:AO}, define $U_c^{\ul}(\ve{q}):=q_c/\set{I}_c^{\ul}(\ve{q})$ and $U_n^{\ul}(\ve{r}):=r_n/\uJn(\ve{r})$, 
		 we can write both problems in the general framework of the max-min fairness problem \eqref{eqn:prob_maxmin_1}:
		\begin{itemize}
		\item[]Problem 1. $\max_{\ve{q}\geq 0}\min_{c\in\set{C}} U_c^{\ul}(\ve{q})/\gamma_c, \|\ve{q}\|\leq \Pm$.
		\item[]Problem 2. $\max_{\ve{r}\geq 0}\min_{n\in\set{N}} U_n^{\ul}(\ve{r}), \|\ve{r}\|\leq \Pm$
		\end{itemize}
		The property of the decoupled variables in uplink and the property of utilities based on the standard interference functions enable us to solve each problem efficiently with two iterative steps: 1) find optimum variable $b_c$ (or $\theta_n$) for each cluster $c$ (or each BS $n$) independently, 2) solve the max-min balancing power allocation problem with fixed-point iteration.
		\subsubsection{Connections between The Two Problems}
		Problem \ref{prob:LB} and Problem \ref{prob:AO} have the same objective $C^{\ul}(\Pm)$ as stated in \eqref{eqn:LB_OP} and \eqref{eqn:maxmin_AO}, i.e., given the same variables $(\hat{\ve{q}}, \hat{\ve{b}}, \hat{\ve{r}}, \hat{\ve{\theta}})$, using \eqref{eqn:AO_utility_1}, we have $\min_{c\in\set{C}} U_c^{\ul}/\gamma_c=\min_{n\in\set{N}} \widehat{U}_n^{\ul}$. Both problems are under the same sum power constraint. However, the convergence of the two-step iteration requires two more properties: 1) the BS power budget $\ve{r}$ derived by solving Problem \ref{prob:AO} at the previous step should not be violated by the cluster power allocation $\ve{q}$ found by optimizing Problem \ref{prob:LB}, and 2) when optimizing Problem \ref{prob:AO}, the inter-cluster power sharing factor $\ve{\beta}$ should be consistent with the derived cluster power allocation $\ve{q}$ in Problem \ref{prob:LB}. 
		
		To fulfill the first requirement, we introduce the per BS power constraint $P_n^{\text{max}}$ for Problem \ref{prob:AO} equivalent to the BS power budget $r_n$ in Problem \ref{prob:LB}. We also propose a scaled version of fixed point iteration similar to the one proposed in \cite{nuzman2007contraction} to iteratively scale the cluster power vector and achieve the max-min utility boundary under per BS power budget constraints, as stated below.
		\begin{equation}
		q_c^{(t+1)} =\frac{\gamma_c\set{I}_c^{\ul}(\ve{q}^{(t)})}{\|\ma{B}\ve{\set{I}}^{\ul}(\ve{q}^{(t)}) \oslash {\ve{P}^{\text{max}}}^{(t)}\|_{\infty}} 
		\label{eqn:FP_LB}
		\end{equation}
		where $\oslash$ denote the element-wise division of vectors, $\|\cdot\|_{\infty}$ denotes the maximum norm, ${\ve{P}^{\text{max}}}^{(t)}:=\ve{r}^{(t)}$. 
		To fulfill the second requirement, once $\ve{q}^{(n+1)}$ is derived, the power sharing factors $\ve{\beta}$ need to be updated for solving Problem \ref{prob:AO} at the next step, given by
		\begin{equation}
		\ve{\beta}^{(n+1)}:=\ma{Q}^{-1}\ma{B}^T\ve{r}^{(n)}, \mbox{where } \ma{Q}=\diag\{\ve{q}^{(n+1)}\}
		\label{eqn:FP_LB_beta}
		\end{equation}
		%
		The scaled fixed-point iteration to optimize Problem \ref{prob:AO} is provided by
		\begin{equation}
		r_n^{(t+1)}= \frac{P^{\text{max}}}{\|\ve{\widehat{\set{I}}}^{\ul}(\ve{r}^{(t)})\|}\cdot \widehat{\set{I}}_n^{\ul}(\ve{r}^{(t)})
		\label{eqn:FP_AO_1}
		\end{equation}
		%
		The joint optimization algorithm is given in Algorithm \ref{alg:optim-algor}.
		\begin{algorithm}[t]\label{alg:optim-algor}
		\caption{Joint Optimization of Problem \ref{prob:LB} and \ref{prob:AO}}
		\begin{algorithmic}[1]
		  \STATE broadcast the information required for computing $\ma{V}$, predefined constraint $P^{\text{max}}$ and thresholds $\epsilon_1,\epsilon_2,\epsilon_3$ 
		  \STATE arbitrary initial power vector $\ve{q}^{(t)}>0$ and iteration step $t:=0$
			\REPEAT[joint optimization of Problem \ref{prob:LB} and \ref{prob:AO}]
		  \REPEAT[fixed-point iteration for every cluster $c\in\set{C}$]
			\STATE broadcast $\ve{q}^{(t)}$ to all base stations
			\FOR{all assignment options $b_c \in \set{N}$}
			\STATE compute $\set{I}_c^{\ul}(\ve{q}^{(t)}, b_c)$ with \eqref{eqn:LB_utility_2}  
			\ENDFOR
			\STATE compute $\set{I}_c^{\ul}(\ve{q}^{(t)})$ with \eqref{eqn:modi_inter_1} and update $b_c^{(t+1)}$
		  \STATE update $q_c^{(t+1)}$ with \eqref{eqn:FP_LB}
		  \STATE $t := t+1$
		  \UNTIL{convergence: $\bigl| q_c^{(t+1)}  - q_c^{(t)}\bigr| / q_c^{(t)} \leq \epsilon_1$}
			\STATE update $\ve{\beta}^{(t)}$ with \eqref{eqn:FP_LB_beta}
			\REPEAT[fixed-point iteration for every BS $n\in\set{N}$]
			\STATE broadcast $\ve{r}^{(t)}$ to all base stations
			\FOR{all antenna tilt options $\theta_n \in \Theta$}
			\STATE compute $\widehat{\set{I}}_n^{\ul}(\ve{r}^{(t)}, \theta_n)$ with \eqref{eqn:AO_utility_2}  
			\ENDFOR
			 \STATE compute $\widehat{\set{I}}_n^{\ul}(\ve{r}^{(t)})$ with \eqref{eqn:modi_inter_2} and update $\theta_n^{(t+1)}$
			 \STATE update $r_c^{(n+1)}$ with \eqref{eqn:FP_AO_1}
			 \STATE $t := t+1$
			 \UNTIL{convergence: $\bigl| r_n^{(t+1)}  - r_n^{(t)}\bigr| / r_n^{(t)} \leq \epsilon_2$}
			\STATE update ${P_n^{\text{max}}}^{(t)}:=r_n^{(t)}$ 
			\STATE compute $l^{(t+1)}:=\min_{n\in\set{N}} \widehat{U}^{\ul}_n(\ve{r}^{(n+1)})$
		\UNTIL{convergence: $|l^{(t+1)}-l^{(t)}|/l^{(t)}\leq\epsilon_3$}
		\end{algorithmic}
		\end{algorithm}
		%
		\section{Uplink-Downlink Duality}\label{sec:Duality}
		
		We state the joint optimization problem in uplink in Section
		\ref{sec:ProbForm} and propose an efficient solution in Section
		\ref{sec:OPAlgor} by exploiting the decoupled property of $\ma{V}$
		over the variables $\ve{\theta}$ and $\ve{b}$. The downlink problem,
		due to the coupled structure of $\ma{V}^T$, is more difficult to
		solve. As extended discussion we want to address the relationship
		between the uplink and the downlink problem, and to propose a
		sub-optimal solution for downlink which can be possibly found through
		the uplink solution.
		
		Let us consider cluster-based max-min capacity utility balancing
		problem in Section \ref{subsubsec:LB_A} as an example. In the downlink
		the optimization problem is written as
		\begin{align}
		\vspace{-0.2em}
		\max_{\ve{q}, \ve{b}}\min_c &\frac{U_c^{(\text{d},1)}(\ve{q}, \ve{b})}{\gamma_c}, \mbox{s.t. } \|\ve{q}\|_1\leq P^{\text{max}}\nonumber\\
		\mbox{where }  & U_c^{(\text{d},1)} :=\frac{q_c}{[\ma{\Psi}\ma{A}\Vb^T\Aa^T\ve{q}+\ma{\Psi}\ve{z}^{\dl}]}
		\label{eqn:LB_dl}
		\vspace{-0.2em}
		\end{align}
		The cluster-based received noise is written as $\ve{z}^{\dl}:=\ma{A}\ve{\sigma}^{\dl}$.
		
		In the following we present a virtual dual uplink network in terms of
		the feasible utility region for the downlink network in
		\eqref{eqn:LB_dl} via Perron-Frobenius theory, such that the solution
		of problem \eqref{eqn:LB_dl} can be derived by solving the uplink
		problem \eqref{eqn:LB_ul} with the algorithm introduced in Section
		\ref{sec:OPAlgor}.
		\begin{proposition}
		  Define a virtual uplink network where the link gain matrix is
		  modified as
		  $\ma{W}_{\ve{b}}:=\diag\{\ve{\alpha}\}\Vb\diag^{-1}\{\ve{\alpha}\}$,
		  i.e., $w_{lk}:=v_{lk}\frac{\alpha_l}{\alpha_k}$, and the received
		  uplink noise is denoted by $\ve{\sigma}^{\ul}:=[{\sigma^2_1}^{\ul},
		  \ldots, {\sigma^2_K}^{\ul}]^T$, where
		  ${\sigma_k^2}^{\ul}:=\frac{\Sigma_{\text{tot}}}{|\set{K}_{c_k}|\cdot
		    C}$ for $k\in\set{K}$, and assume
		  $\Sigma_{\text{tot}}:=\|\ve{\sigma}^{\ul}\|_1=\|\ve{\sigma}^{\dl}\|_1$
		  (which means, the sum noise is equally distributed in clusters,
		  while in each cluster the noise is equally distributed in the
		  subordinate users). The dual uplink problem of problem
		  \eqref{eqn:LB_dl} is given by
		\begin{align}
		\vspace{-0.2em}
		\max_{\ve{q},\ve{b}}\min_c & \frac{U_c^{(\text{u},1)}(\ve{q}, \ve{b})}{\gamma_c}, \mbox{s.t. } \|\ve{q}\|_1\leq \Pm\nonumber\\
		\mbox{where } & U_c^{(\text{u},1) }:=\frac{q_c}{[\ma{\Psi}\ma{A}\ma{W}_{\ve{b}}\Aa^T\ve{q}+\ma{\Psi}\ve{z}^{\ul}]}
		\label{eqn:LB_ul}
		\vspace{-0.2em}
		\end{align}
		where $\ve{z}^{\ul}:=\ma{A}\ve{\sigma}^{\ul}$.
		\label{prop:Duality}
		\end{proposition}
		\begin{proof} The proof is given in the Appendix.
		\end{proof}
		
		Note that the optimizer $\ve{b}^{\ast}$ for BS assignment in downlink can be equivalently found by minimizing the spectral radius $\ma{\Lambda^{(u)}(\ve{b})}$ in the uplink. Once $\ve{b}^{\ast}$ is found, the associate optimizer for uplink power ${\ve{q}^{\ul}}^{\ast}$ is given as the dominant right-hand eigenvector of matrix $\ma{\Lambda}^{\ul}(\ve{b}^{\ast})$, while the associate optimizer for downlink power ${\ve{q}^{\dl}}^{\ast}$ is given as the dominant right-hand eigenvector of matrix $\ma{\Lambda}^{\dl}(\ve{b}^{\ast})$. 
		Proposition \ref{prop:Duality} provides an efficient approach to solve
		the downlink problem with two iterative steps (as the one proposed in
		\cite{BocheDuality06}): 1) for a fixed power allocation
		$\hat{\ve{q}}$, solve the uplink problem and derive the assignment
		$\ve{b}^{\ast}$ that associated with the spectral radius of extend
		coupling matrix $\ma{\Lambda}^{\ul}$, and 2) for a fixed assignment
		$\hat{\ve{b}}$, update the power $\ve{q}^{\ast}$ as the solution of
		\eqref{eqn:DL_matrixEqua}. 
		
		Although we are able to find a dual uplink problem for the
		downlink problem in \eqref{eqn:LB_dl} with our proposed utility
		functions \emph{under sum power
		  constraints}, 
		\insl{we are not able to construct a dual network with decoupled properties for the modified problem
		\emph{under per BS power constraints} \eqref{eqn:FP_LB}. However,
		numerical experiments show that our approach to the downlink
		based on the proposed uplink solution does improve the network
		performance, although the duality does not exactly hold between the downlink problem and our proposed uplink problem under the per BS power constraints.}
		%
		%
		\section{Numerical Results}\label{sec:Simu}
 We consider a real-world urban scenario based on a pixel-based mobility model of realistic collection of BS locations and pathloss model for the city of Berlin. The data was assembled within the EU project MOMENTUM and is available at  \cite{MOMENTUM}. We select 15 tri-sectored BS in the downtown area. Users are uniformly distributed and are clustered based on their SINR distributions as shown in Fig. \ref{fig:Berlin} (UEs assigned to each sector are clustered into groups and are depicted in distinct colors). The SINR threshold is defined as -6.5 dB and the power constraint per BS is 46dBm. The 3GPP antenna model defined in \cite{3GPP36942} is applied.    
 		
		Fig. \ref{fig:convergence} illustrates the convergence of the algorithm.  Our algorithm achieves the max-min utility balancing, and improves the feasibility level $C^{(u)}(\Pm)$ by each iteration step. 
		
		In Fig.\ref{fig:cov_cap_mu} we show that the trade-off between coverage and capacity can be adjusted by tuning parameter $\mu$. By increasing $\mu$ we give higher priority to capacity utility (which is proportional to the ratio between total useful power and total interference power), while for better coverage utility (defined as minimum of SINRs) we can use a small value of $\mu$ instead.
		
		Fig. \ref{fig:coverage}, \ref{fig:capacity} and \ref{fig:power} illustrate the improvement of coverage and capacity performance and decreasing of the energy consumption in both uplink and downlink systems by applying the proposed algorithm, when the average number of the users per BS is chosen from the set $\{15,20,25,30,35\}$. In Fig. \ref{fig:capacity} we show that the actual average SINR is also improved, although the capacity utility is defined as a lower bound of the average SINR. Fig. \ref{fig:power} illustrate that our algorithm is more energy efficient when comparing with the fixed BS power budget scenario. Compared to the near-optimal uplink solutions, less improvements are observed for the downlink solutions as shown in Fig. \ref{fig:coverage}, \ref{fig:capacity} and \ref{fig:power}. This is because we derive the downlink solution by exploiting an uplink problem which is not exactly its dual due to the individual power constraints (as described in Section \ref{sec:Duality}). However, the sub-optimal solutions still provide significant performance improvements.
		\section{Conclusions and Further Research}\label{sec:con}
		We present an efficient and robust algorithmic optimization framework build on the utility model for joint optimization of the SON use cases coverage and capacity optimization and load balancing. The max-min utility balancing formulation is employed to enforce the fairness across clusters. We propose a two-step optimization algorithm in the uplink based on fixed-point iteration to iteratively optimize the per base station antenna tilt and power allocation as well as the cluster-based BS assignment and power allocation. We then analyze the network duality via Perron-Frobenius theory, and propose a sub-optimal solution in the downlink by exploiting the solution in the uplink. Simulation results show significant improvements in performance of coverage, capacity and load balancing in a power-efficient way, in both uplink and downlink. In our follow-up papers we will further propose a more complex interference coupling model and the optimization framework where frequency band assignment is taken into account. We will also examine the suboptimality under more general form of power constraints.   
		 
		 	\begin{figure}[t]
		  \centering
		  \includegraphics[width=.5\textwidth]{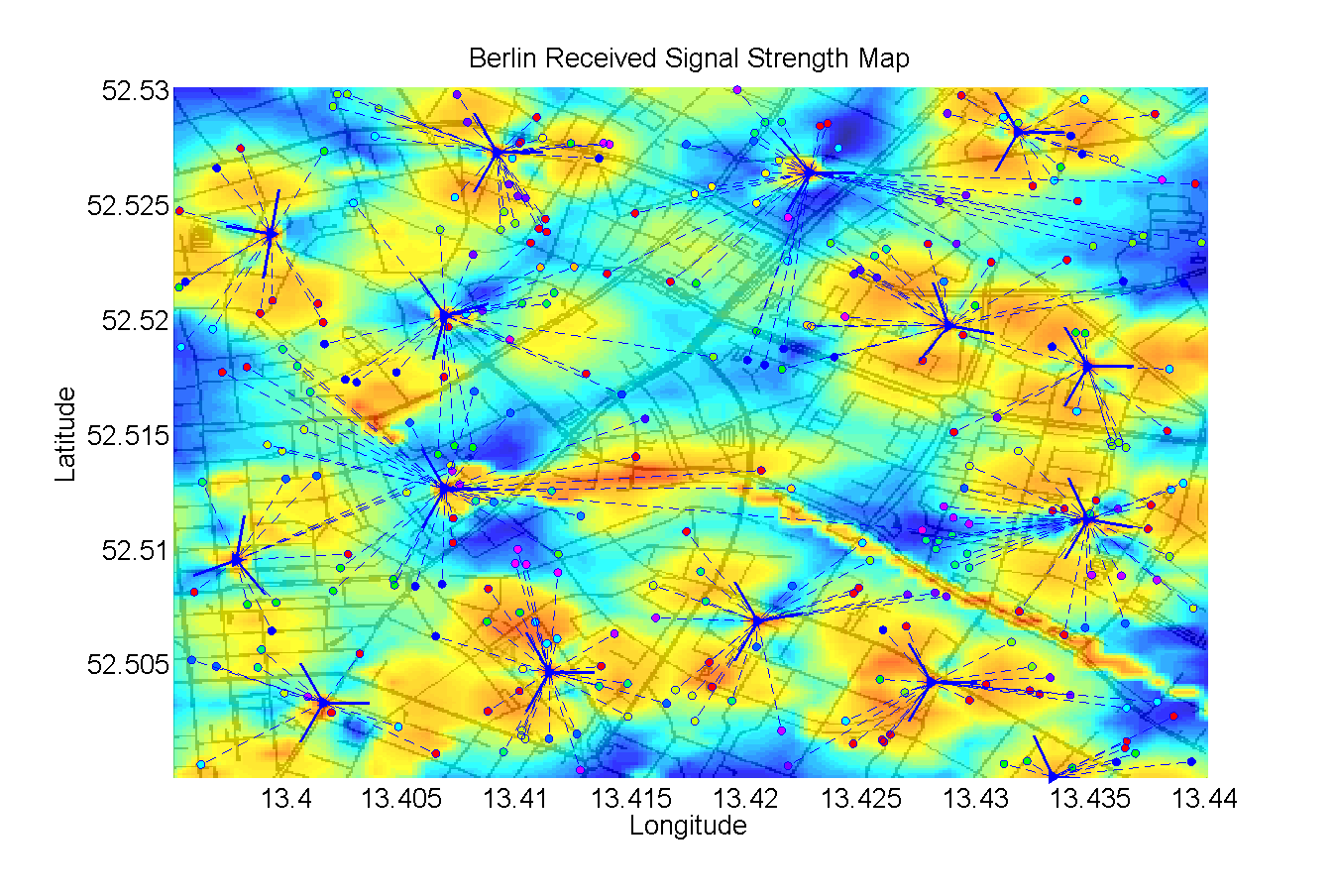}
		  \caption{Berlin Scenario.}
		  \label{fig:Berlin}
		\end{figure}
		
		\begin{figure}[ht]
		  \centering
		  \includegraphics[width=.5\textwidth]{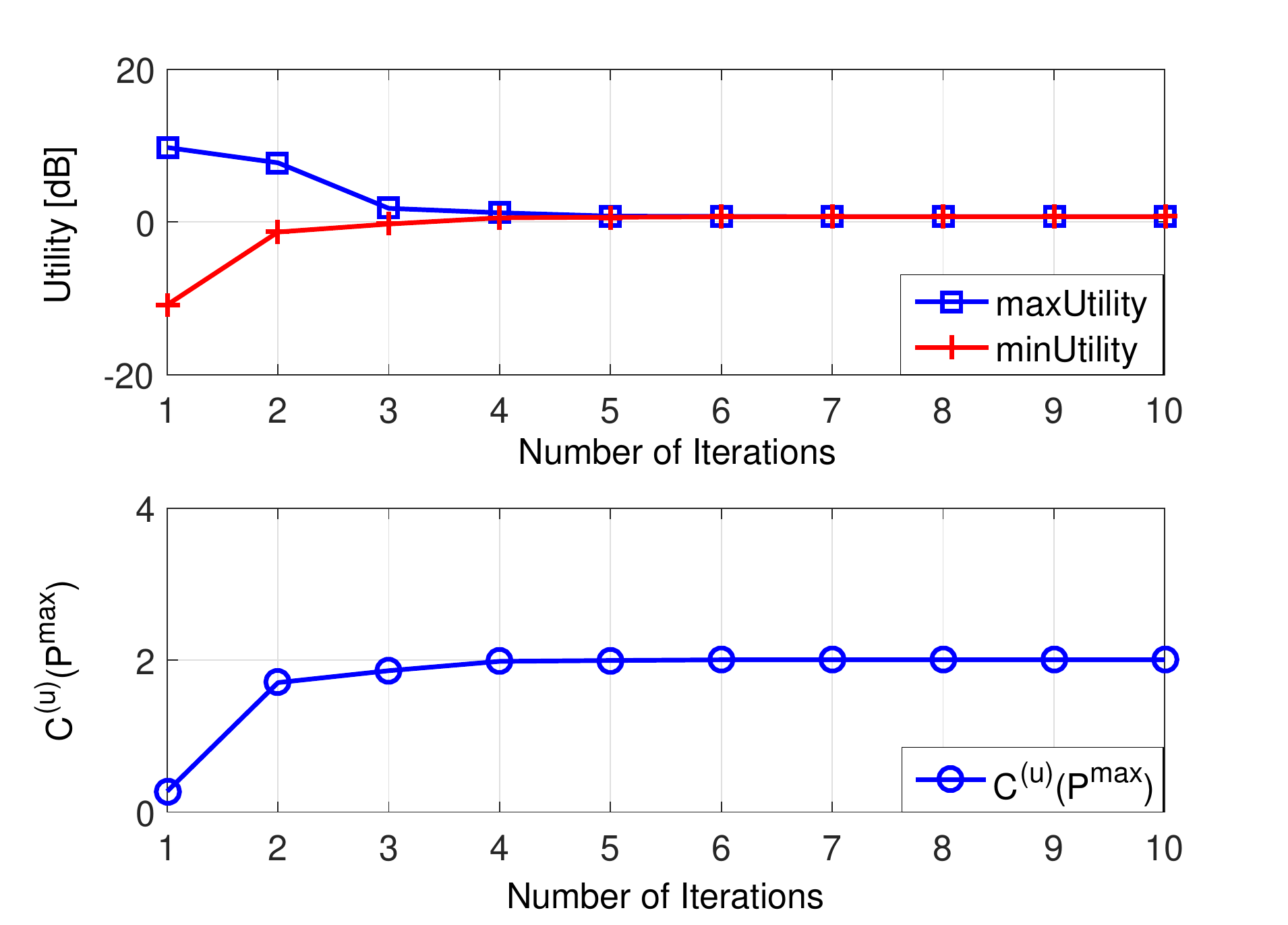}

		  \caption{Algorithm convergence.}
		  \label{fig:convergence}
		\end{figure}
		%
		\begin{figure}[ht]
		  \centering
		  \includegraphics[width=.5\textwidth]{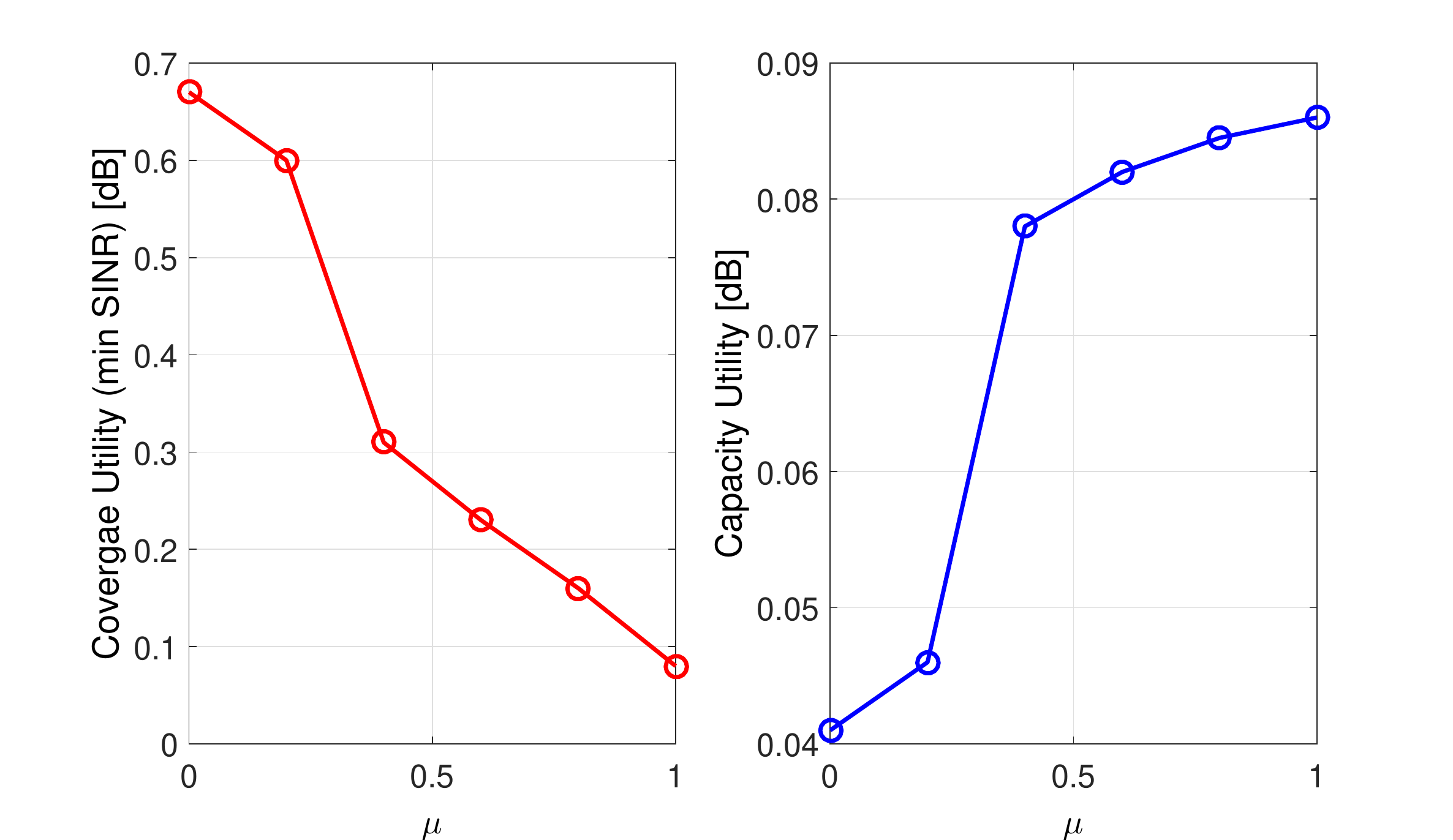}  
			\vspace{-1em}
		  \caption{Trade-off between utilities depending on $\mu$.}
		  \label{fig:cov_cap_mu}
			\vspace{-1.5em}
		\end{figure}
		\begin{figure}[ht]
		  \centering
		  \includegraphics[width=.5\textwidth]{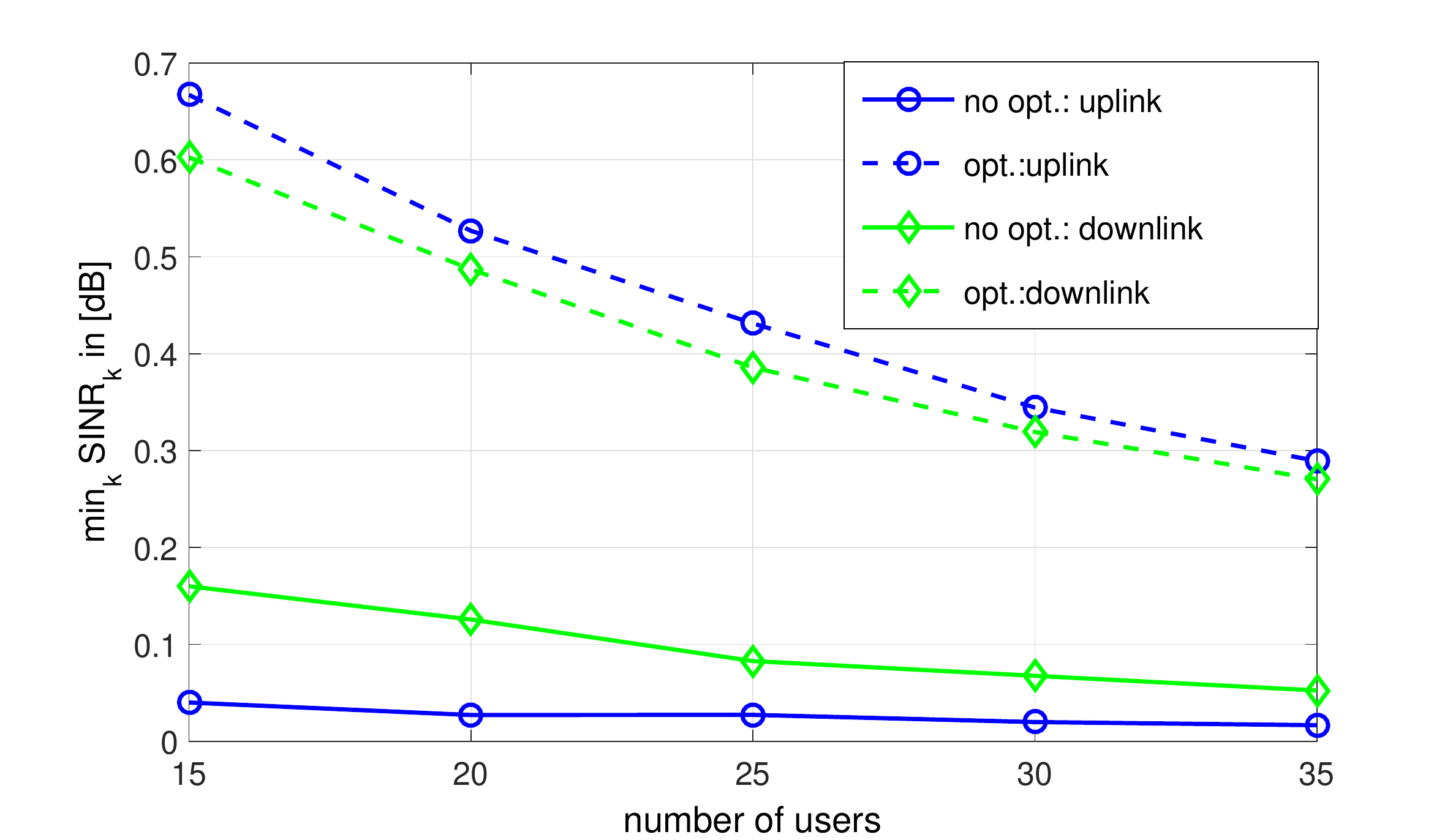} 
		  \caption{Performance of proposed algorithm: coverage.}
		  \label{fig:coverage}
		\end{figure}
		\begin{figure}[ht]
		  \centering
		  \includegraphics[width=.43\textwidth]{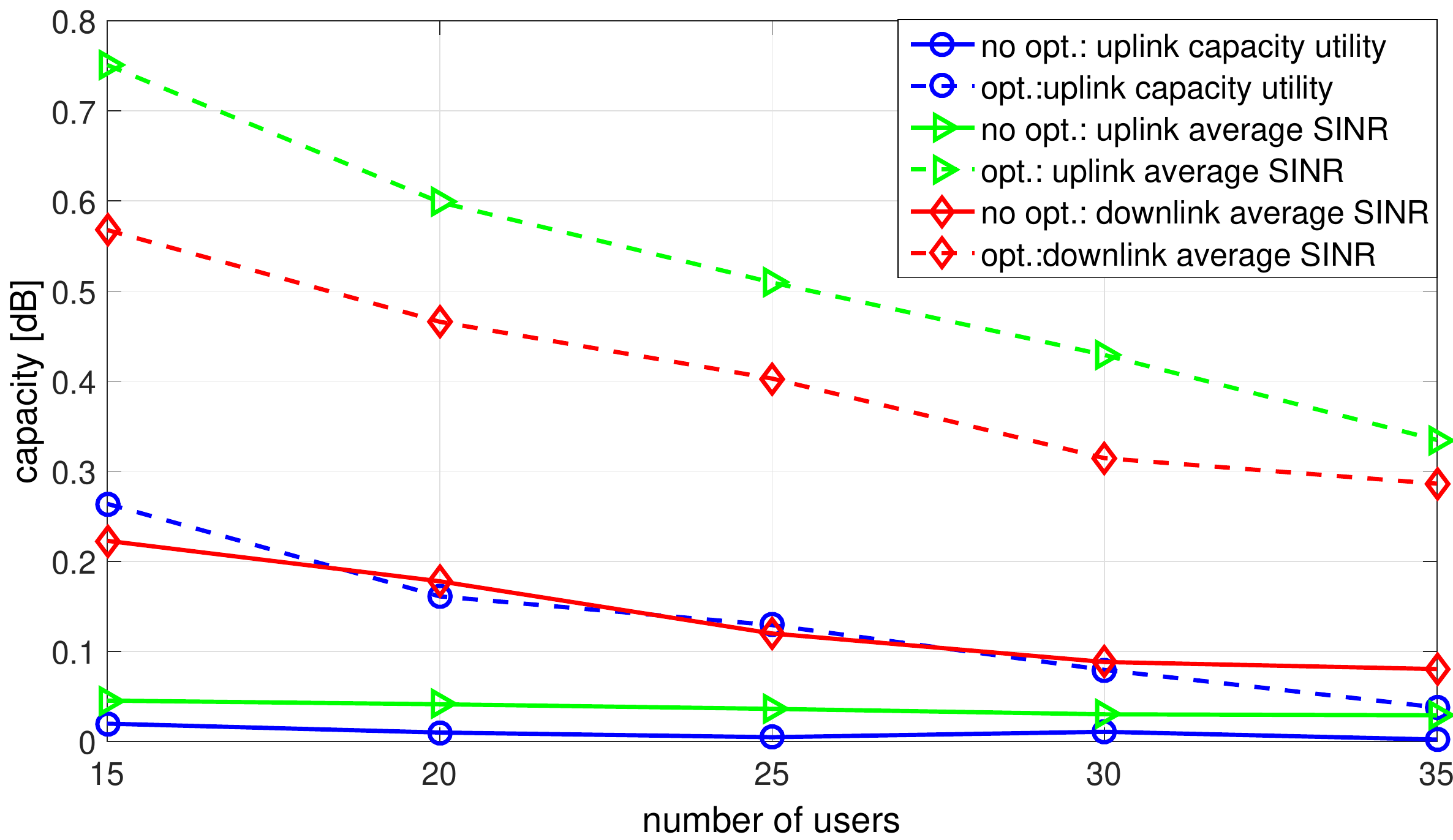} 
		  \caption{Performance of proposed algorithm: capacity.}
		  \label{fig:capacity}
		\end{figure}
		\begin{figure}[!ht]
		  \centering
		  \includegraphics[width=.43\textwidth]{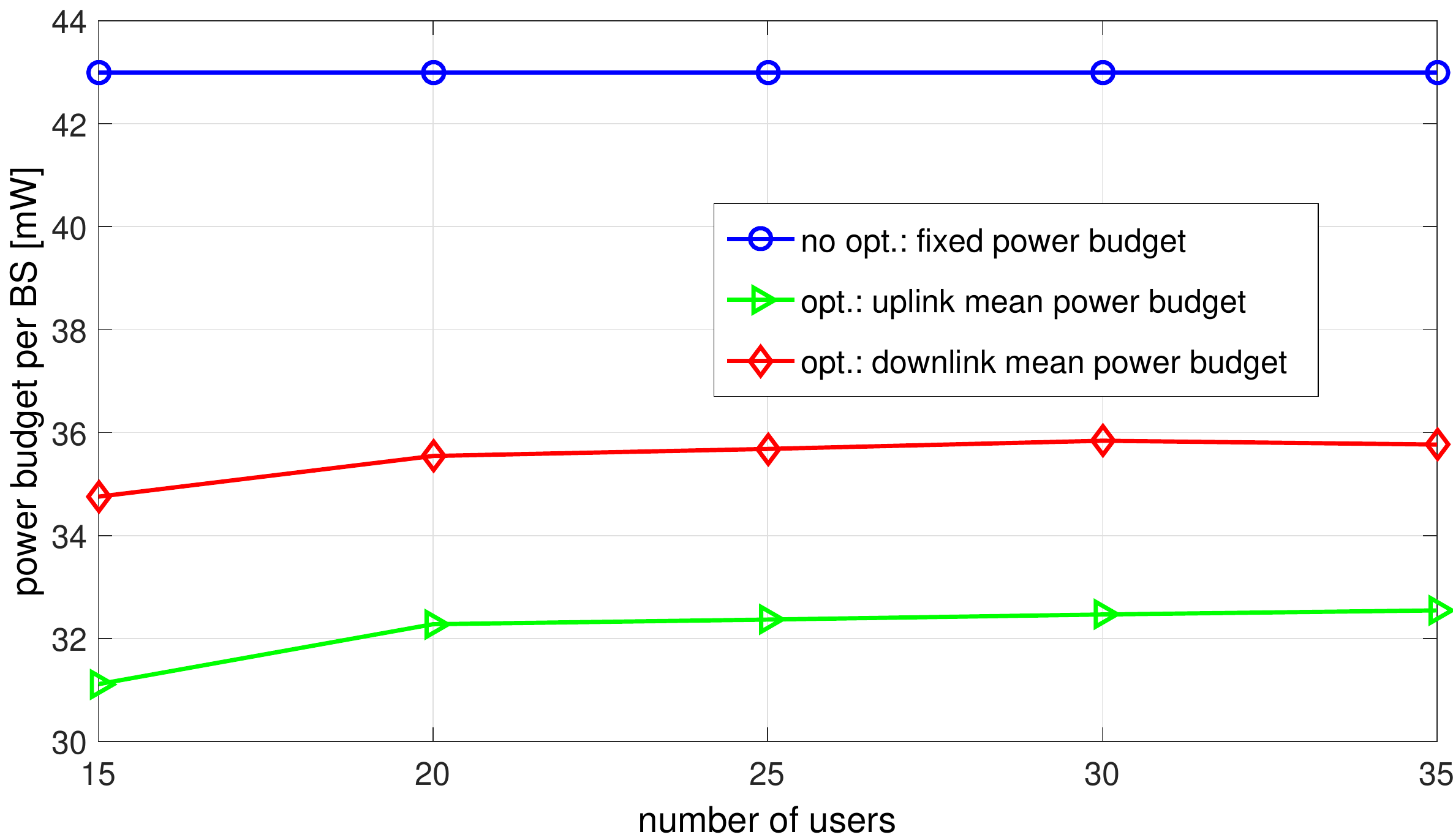} 
		  \caption{Performance of proposed algorithm: per-BS power budget.}
		  \label{fig:power}
		\end{figure}

\appendix
\begin{proof} 

{\it Proposition \ref{prop:Duality}:}
For any fixed BS assignment $\hat{\ve{b}}$, denote $\hW:=\ma{W}_{\hat{\ve{b}}}$ and $\hV:=\Vb$ for convenience,  the optimal downlink power solution $\hat{\ve{q}}^{\dl}$ for problem \eqref{eqn:LB_dl} satisfies \cite{stanczak2009fundamentals}
		\begin{equation}
		\vspace{-0.2em}
		\ma{\Lambda}^{\dl} \hat{\ve{q}}^{\dl}=\frac{1}{C^{\dl}(\hat{\ve{b}},\Pm)} \hat{\ve{q}}^{\dl}, \hat{\ve{q}}^{\dl}\in\R_{+}^C
		\label{eqn:DL_matrixEqua}
		\vspace{-0.2em}
		\end{equation}
		where $\ma{\Lambda}^{\dl}\in\R_{+}^{C\times C}$ is defined as
		\begin{equation}
		\vspace{-0.2em}
		\ma{\Lambda}^{\dl}:=\ma{\Gamma}\ma{\Psi}\left[\ma{A}\hV^T\Aa^T+\frac{1}{\Pm}\ve{z}^{\dl}\ve{1}_C^T\right].
		\label{eqn:DL_Lambda}
		\vspace{-0.2em}
		\end{equation}
		we denote $\ma{\Gamma}:=\diag\{\gamma_1,\ldots, \gamma_C\}$, $C^{\dl}(\hat{\ve{b}},\Pm)=\max_{\ve{q}\geq 0}\min_c U_c^{(\text{d},1)}/\gamma_c$ subject to $\|\ve{q}\|_1\leq \Pm$, and 
		$\ve{1}_C$ is a C-dimensional all-one vector.
		\eqref{eqn:DL_matrixEqua} and \eqref{eqn:DL_Lambda} are derived by writing the utility fairness $U_c^{(\text{d},1)}/\gamma_c=C^{\dl}(\hat{\ve{b}}, \Pm)$ for all $c\in\set{C}$ and the power constraint $\|\ve{q}^{\dl}\|_1=\Pm$ with matrix notation. Targets $\ve{\gamma}$ is feasible if and only if $C^{\dl}(\hat{\ve{b}},\Pm)>1$. 
		
		Similarly, the optimal uplink power solution $\hat{\ve{q}}^{\ul}$ for uplink problem \eqref{eqn:LB_ul} needs to satisfy 
		\begin{equation}
		\vspace{-0.2em}
		\ma{\Lambda}^{\ul}\hat{\ve{q}}^{\ul}=\frac{1}{C^{\ul}(\hat{\ve{b}},\Pm)} \hat{\ve{q}}^{\ul}, \hat{\ve{q}}^{\ul}\in\R_{+}^C
		\label{eqn:UL_matrixEqua}
		\vspace{-0.2em}
		\end{equation}
		where $\ma{\Lambda}^{\ul}\in\R_{+}^{C\times C}$ is defined as
		\begin{equation}
		\vspace{-0.2em}
		\ma{\Lambda}^{\ul}:=\ma{\Gamma}\ma{\Psi}\left[\ma{A}\hW\Aa^T+\frac{1}{\Pm}\ve{z}^{\ul}\ve{1}_C^T\right].
		\label{eqn:UL_Lambda}
		\vspace{-0.2em}
		\end{equation}
		where $\ve{z}^{\ul}:=\ma{A}\ve{\sigma}^{\ul}$, i.e., $z_c^{\ul}=\Sigma_{\text{tot}}/C$ for all $c\in\set{C}$. 
		
		The balanced level $C^{\dl}(\hat{\ve{b}},\Pm)$ and
		$C^{\ul}(\hat{\ve{b}},\Pm)$ are the reciprocal spectral radius of the
		nonnegative extended coupling matrix $\ma{\Lambda}^{\dl}$ and
		$\ma{\Lambda}^{\ul}$. Moreover, according to Perron-Frobenius theorem, if both $\ma{\Lambda}^{\dl}$ and
		$\ma{\Lambda}^{\ul}$ are irreducible, they have unique real spectral radius and their corresponding eigenvectors (power allocation) have strictly positive components. By comparing the interference terms in
		\eqref{eqn:DL_Lambda} and \eqref{eqn:UL_Lambda}, we have
		$(\ma{A}\hV^T\Aa^T)^T=\Aa\hV\ma{A}^T=\ma{A}\diag\{\ve{\alpha}\}\hV
		\ma{I}\ma{A}^T=\ma{A}\diag\{\ve{\alpha}\}\hV\diag^{-1}\{\ve{\alpha}\}\diag\{\ve{\alpha}\}
		\ma{A}^T=\ma{A}\hW^T\Aa^T$. By comparing the noise terms we have
		$\ve{z}^{\ul}=\frac{1}{C} \ve{1}_C{\ve{z}^{\dl}}^T\ve{1}_C$ (by using
		$z_c^{\ul}=\Sigma_{\text{tot}}/C$ for all $c\in\set{C}$), thus
		$\ve{z}^{\ul}\ve{1}_C^T=\frac{1}{C}
		\ve{1}_C{\ve{z}^{\dl}}^T\ve{1}_C\ve{1}_C^T=\ve{1}_C{\ve{z}^{\dl}}^T=(\ve{z}^{\dl}\ve{1}_C^T)^T$. By
		using the properties of spectral radius $\rho(\ma{X})=\rho(\ma{X}^T)$
		and $\rho(\ma{X}\ma{Y})=\rho(\ma{Y}\ma{X})$ we have that
		$\rho(\ma{\Lambda}^{\dl})=\rho(\ma{\Lambda}^{\ul})$ and thus
		$C^{\dl}(\hat{\ve{b}},P^{\text{max}})=C^{\ul}(\hat{\ve{b}},P^{\text{max}})$. Notice
		that the network duality holds for any given BS assignment
		$\hat{\ve{b}}$, the achievable utility regions are the same for both
		the downlink problem \eqref{eqn:LB_dl} and uplink problem
		\eqref{eqn:LB_ul}.
		\end{proof}
		\subsection*{Acknowledgements} 
		We would like to thank Dr. Martin Schubert and Dr. Carl J. Nuzman for their expert advice.
		\ifCLASSOPTIONcaptionsoff
		  \newpage
		\fi
		\bibliographystyle{IEEEtran}
		\bibliography{main}
	\end{document}